\documentclass{article}

\usepackage{latexsym,amsmath,amssymb,url}
\usepackage{amsthm}

\DeclareMathOperator*{\argmax}{arg\,max}

\newtheorem{lemma}{Lemma}

\newtheorem{definition}{Definition}
\newtheorem{theorem}{Theorem}
\usepackage[ruled, linesnumbered]{algorithm2e}

\newcommand{\edgekerneltwo}{$8k$}
\newcommand{\edgekernel}{$16c^{2}k-6ck$}
\newcommand{\approxtwo}{$4+\varepsilon$}

\title{Parameterized and Approximation Algorithms for the Load Coloring Problem}
\author{F. Barbero, G. Gutin, M. Jones, and B. Sheng\\ {\small Royal Holloway, University of London}\\{\small TW20 0EX, Egham, Surrey, UK}}
\date{\today}

\begin{document}
\maketitle

\begin{abstract}
Let $c, k$ be two positive integers and let $G=(V,E)$ be a graph.
The $(c,k)$-Load Coloring Problem (denoted $(c,k)$-LCP) asks whether there is a $c$-coloring $\varphi: V \rightarrow [c]$ such that for every $i \in [c]$, there are at least $k$ edges with both endvertices colored $i$. 
Gutin and Jones (IPL 2014) studied this problem with $c=2$. They showed $(2,k)$-LCP to be fixed parameter tractable (FPT) with parameter $k$ by obtaining a kernel with at most $7k$ vertices. In this paper, we extend the study to any fixed $c$ by giving both a linear-vertex and a linear-edge kernel. In the particular case of $c=2$, we obtain a  kernel with less than $4k$ vertices and less than \edgekerneltwo ~edges. These results imply that for any fixed $c\ge 2$, $(c,k)$-LCP is FPT and that the optimization version of $(c,k)$-LCP (where $k$ is to be maximized) has an approximation algorithm with a constant ratio for any fixed $c\ge 2$. 
\end{abstract}

\section{Introduction}
Given a graph $G=(V,E)$ and an integer $k$, the 2-{\sc Load Coloring} Problem asks whether there is a coloring $\varphi: V \rightarrow \{1,2\}$ such that for $i=1$ and 2, there are at least $k$ edges with both endvertices colored $i$. This problem is NP-complete \cite{NABA2007}, and Gutin and Jones studied its parameterization by $k$ \cite{GJ2013}. They proved that 2-{\sc Load Coloring} is fixed-parameter tractable by obtaining a kernel with at most $7k$ vertices. It is natural to extend 2-{\sc Load Coloring} to any number $c$ of colors as follows. Henceforth, for a positive integer $p$, $[p]=\{1,2,\dots ,p\}.$

\begin{definition}[$(c,k)$-{\sc Load Coloring}]
Given a positive integer $c$, a nonnegative integer $k$ and graph $G=(V,E)$,
the $(c,k)$-{\sc Load Coloring} Problem asks whether there is a $c$-coloring $\varphi: V\rightarrow [c]$ such
that for every $i \in [c]$, there are at least $k$ edges with both endvertices colored $i$.
We write $G\in (c,k)$-LCP if such a $c$-coloring exists.
\end{definition}
 Observe first that $G\in (1,k)$-LCP if and only if $|E(G)| \ge k$. In this paper, we consider $(c,k)$-{\sc Load Coloring} parameterized by $k$ for every fixed $c\ge 2$. Note that $(c,k)$-{\sc Load Coloring} is NP-complete for every fixed $c\ge 2$. Indeed, we can reduce $(2,k)$-{\sc Load Coloring} to $(c,k)$-{\sc Load Coloring} with $c>2$ by taking the disjoint union of $G$ with $c-2$ stars $K_{1,k}$. 

 We prove that the problem admits a kernel with less than $2ck$ vertices. Thus, for $c=2$ we improve the kernel result of \cite{GJ2013}. To show our result, we introduce reduction rules, which are new even for $c=2$. We prove that the reduction rules can run in polynomial time and we show that a reduced graph with at least $2ck$ vertices is in $(c,k)$-LCP. 

While there are many parameterized graph problems which admit kernels linear in the number of vertices, usually only problems on classes of sparce graphs admit kernels linear in the number of edges (since in such graphs the number of edges is linear in the number of vertices), see, e.g., \cite{meta,DF2013,LMS2012}. To the best of our knowledge, only trivial $O(k)$-edge kernels for general graphs have been described in the literature, e.g., the kernel for {\sc Max Cut} parameterized by solution size. Thus, our next result is somewhat surprising: $(c,k)$-{\sc Load Coloring} admits a kernel with $O(k)$ edges for every fixed $c\ge 2$. In fact, $(2,k)$-{\sc Load Coloring} has a kernel with less than \edgekerneltwo ~edges and for every $c\ge 2$, $(c,k)$-{\sc Load Coloring} has a kernel with less than \edgekernel ~edges. 

The optimization version of $(c,k)$-{\sc Load cColoring}, called the {\sc $c$-Load Coloring} Problem, is as follows: for a graph $G$ and an integer $c\ge 2$, find the maximum $k$ such that $G\in (c,k)$-LCP.  The above bounds on the number of edges in the kernel lead to approximation algorithms for this optimization problem:
a (\approxtwo)-approximation for $c=2$ and a constant ratio approximation for $c>2$. 

The paper is organized as follows. In Section \ref{sec:tn}, we provide additional terminology and notation. In Section \ref{sec:vk}, we show that the problem admits a kernel with less than $2ck$ vertices. In Section \ref{sec:ek}, we prove an upper bound on the number of edges in a kernel for every $c\ge 2$ and the corresponding approximation result for {\sc $c$-Load Coloring}. 
We improve our bound for $c=2$ in Section \ref{sec:ek2}. The bound implies the approximation ratio of \approxtwo ~for every $\varepsilon >0$.
We complete the paper with discussions in Section \ref{sec:d}.

\section{Terminology and Notation}\label{sec:tn}
\paragraph{Graphs.} For a graph $G$, $V(G)$ ($E(G)$, respectively) denotes the vertex (edge, respectively) set of $G$, 
$\Delta(G)$ denotes the maximum degree of $G$, $n$ its number of vertices, and $m$ its number of edges. 
For a vertex $x$ and vertex set $X$ in $G$, $N(x)=\{y: xy\in E(G)\}$ and $N_X(x) = N(x) \cap X$.
For disjoint vertex sets $X,Y$ of $G$, let $G[X]$ be the subgraph of $G$ induced by $X$, $E(X)=E(G[X])$ and $E(X,Y)=\{xy\in E(G):\ x\in X, y\in Y\}.$
A vertex $u$ with degree 0 (1, respectively) is an \textit{isolated vertex} (a \textit{leaf-neighbor of $v$}, where $uv\in E(G),$ respectively).
For a coloring $\varphi$, we say that an edge $uv$ is \textit{colored $i$} if $\varphi(u)=\varphi(v)=i$.\\

\paragraph{Parameterized complexity.}A parameterized problem is a subset $L\subseteq \sum^{*}\times \mathbb{N}$ over a finite alphabet $\sum$. $L$ is {\em fixed-parameter tractable} (FPT) if the membership of an instance $(x,k)$ in $\sum^{*}\times \mathbb{N}$ can be decided in time $f(k)|x|^{O(1)}$, where $f$ is a computable function of the parameter $k$ only. A $\textit{kernelization}$ of a parameterized problem $L$ is a polynomial-time algorithm that maps an instance $(x, k)$ to an instance $(x', k')$, the \textit{kernel}, such that $(x, k)\in L$ if and only if $(x', k') \in L$, $k' \le g(k)$, and $|x'| \leq g(k)$ for some function $g$ of $k$ only. 
We call $g(k)$ the {\em size} of the kernel.\\

It is well-known that a parameterized problem $L$ is FPT if and only if it is decidable and admits a kernelization. Due to applications, low degree polynomial size kernels are of main interest. Unfortunately, many FPT problems do not have kernels of polynomial size unless the polynomial hierarchy collapses to the third level \cite{CFLMPPS2015,DF2013}. For further background and terminology on parameterized complexity we refer the reader to the monographs \cite{CFLMPPS2015,DF2013,JM2006,R2006}.

\section{Bounding Number of Vertices in Kernel}\label{sec:vk}

In this section, we show that $(c,k)$-{\sc Load Coloring} admits a kernel with less than $2ck$ vertices.
A matching with $2ck$ vertices 
suggests that this bound is likely to
be optimal.

For $\tau \in \{<,\le,=, >,\ge\}$ and integer $i\ge 1$, $K_{1,\tau i}$ denotes a star $K_{1,j}$ with $j~\tau~i$ and $j\ge1$. For example, $K_{1,\le p}$ is a star with $q$ edges such that $q\in [p]$. A \textit{$K_{1,\tau i}$-graph} is a forest in which every component is a star $K_{1,\tau i}$, and a \textit{$K_{1,\tau i}$-cover} of $G$ is a $K_{1,\tau i}$-subgraph $F$ of $G$ such that $V(F)=V(G)$. We call any $K_{1,\tau i}$-graph a \textit{star graph} and any $K_{1,\tau i}$-cover a \textit{star cover}.

We first prove the bound for star graphs with small maximum degree.

\begin{lemma}\label{lem:starkernel} If $G$ is a $K_{1,<2k}$-graph with $n\geq 2ck$, then $G \in (c,k)$-LCP.
\end{lemma}

\begin{proof} The idea is to find for each color some induced subgraph with at least $k$ edges and at most $2k$ vertices. If such subgraphs exist, it is possible to color at most $2ck$ vertices of the graph to obtain $k$ edges for each of the $c$ colors. We prove the lemma by induction on $c$. The base case of $c=1$ holds since a $K_{1,<2k}$-graph $G$ with at least $2k$ vertices has at least $k$ edges 
(observe that a $K_{1,<2k}$-graph has no isolated vertices).

Observe now that because all components of $G$ are trees, for each one the number of vertices is one more than the number of edges.
If there is a component $C,$ with $k \leq |E(C)| < 2k$, color $V(C)$ with the same color. Then we have used $|V(C)|\leq 2k$ vertices. Thus, we may assume that   
every component has less than $k$ edges and let $C_{1}, C_{2},\ldots , C_{t}$ be the components of $G$. Let $b$ be the minimum nonnegative integer for which there exists $I\subseteq [t]$ such that $\Sigma_{i\in I}|E(C_{i})|=k+b\geq k$. Since there is no isolated vertex in a star graph, $m \ge n/2 \ge ck$, and thus such a set $I$ exists. Observe that for any $i \in I$, $|E(C_{i})|> b$, as otherwise $\Sigma_{j\in I\backslash\{i\}}|E(C_{j})|=k+b-|E(C_{i})|\geq k$, a contradiction to the minimality of $b$. Since every component has less than $k$ edges, $b\leq k-2$.

For a star $(V,E)$, the ratio $\frac{|V|}{|E|}$ decreases when $|E|$ increases. Thus, we have $\Sigma_{j\in I}|V(C_{j})|\leq (k+b)\frac{b+2}{b+1}$. But $2k-(k+b)\frac{b+2}{b+1}=\frac{(k-2-b)b}{b+1}\geq 0$, and so $\Sigma_{j\in I}|V(C_{j})|\leq 2k$. We may color the components $C_{i}$, $i\in I$, by the same color. Observe that $H=G-V(\bigcup_{i\in I} C_{i})$ has at least $2(c-1)k$ vertices and so $H\in (c-1,k)$-LCP by the induction hypothesis. 
Thus, $G\in (c,k)$-LCP.
\end{proof}

Since $G \in (c,k)$-LCP  whenever $G$ has a subgraph $H \in (c,k)$-LCP, we have that any graph with $n \geq 2ck$ and a $K_{1,<2k}$-cover is in $(c,k)$-LCP.\\

We introduce now a family $(O_{i,k})_{i,k \in \mathbb{N}}$ of obstacles.

\begin{definition} We call a pair $(V_1,V_2)$ of disjoint vertex sets an {\em obstacle from $O_{i,k}$} if $|V_1|=i$, $ N(v) \subseteq V_1$ for all $v \in V_2$, and
for every $u \in V_1$ there is a set $V_u \subseteq N_{V_2}(u)$ such that $|V_u| \geq k$ and for every pair $u,v$ of distinct vertices of $V_1$, $V_u \cap V_v = \emptyset$.
\end{definition}

Note that if $v$ is an isolated vertex, the pair $(\emptyset, \{v\})$ is an obstacle from $O_{0,k}$.

Observe that if an obstacle ($V_1$, $V_2$) from $O_{i,k}$ is contained in a graph $G$, then $G[V_{1}\cup V_{2}] \in (i,k)$-LCP: color each $u \in V_1$ and $V_u$ with one color. However,  $G[V_{1}\cup V_{2}]\notin (i+1,k)$-LCP. Indeed, every edge in $G[V_{1}\cup V_{2}]$ is incident to at least one of the $i$ vertices in $V_1$. Thus, an edge can only be colored with one of $|V_1| = i$ colors. From this observation, we deduce the following set of reduction rules.

\textbf{Reduction~rule~$\mathbf{R_{i,k}}$.} If an instance $G$ for $(c,k)$-LCP contains an obstacle $(V_1,V_2)$ from $O_{i,k}$, delete all the vertices of $V_1 \cup V_2$ and decrease $c$ by $i$.

\medskip

Now we will prove that Rules $R_{i,k}$ are safe and can be applied in time polynomial in $n$ (recall that $c$ is fixed).

\begin{lemma}\label{lem2}
Let $G$ be a graph and $G'$ be the graph obtained from $G$ after applying reduction rule $R_{i,k}$. Then $G\in (c,k)$-LCP if and only if $G'\in (c-i,k)$-LCP. 
\end{lemma}

\begin{proof}
For a positive integer $p$, we call a coloring of an instance $G$ of $(c,k)$-LCP a \textit{good coloring with $p$ colors} if for at least $p$ colors $j \in [c]$, there are at least $k$ edges colored with color $j$.

If $G'\in (c-i,k)$-LCP, then $G\in (c,k)$-LCP, since a good coloring of the obstacle with $i$ colors together with a good coloring of $G'$ with $c-i$ colors gives a good coloring of $G$ with $c$ colors. On the other hand, if $G\in (c,k)$-LCP, then it has a good coloring with $c$ colors. In this coloring, there are at least $c-|V_1| = c-i$ colors with no edge with endvertices in $V_1$. These colors must have their $k$ edges in $E(G-V_1) = E(G')$. Thus $G'\in (c-i,k)$-LCP.
\end{proof}

\begin{lemma}\label{poly}
One can decide whether Rule $R_{i,k}$ is applicable to $G$ in time $O(n^{i+O(1)})$. 
\end{lemma}
\begin{proof}
Generate all $i$-size subsets $V_1$ of $V(G)$. For each $V_1$, construct the set $V_2$ that includes every vertex outside $V_1$ whose only neighbors are in $V_1$. If $|V_2| \geq ik$, construct the following bipartite graph $B$: the partite sets of $B$ are $V'_1$ and $V_2$, where $V'_1$ contains $i$ copies of every vertex $v$ of $V_1$ with the same neighbors as $v$. Observe that $B$ has a matching covering $V'_1$ if and only if $R_{i,k}$ can be applied to $G$ for the obstacle $(V_1,V_2)$. It is not hard to turn the above into an algorithm of runtime $O(n^{i+O(1)})$.
\end{proof}

We say that a graph is \textit{reduced for $(c, k)$-LCP} if it is not possible to apply any rule $R_{i,k}$, $i < c$ to the graph. 

\begin{lemma}\label{lem:coverExists}
Let $G$ be a reduced graph for $(c,k)$-LCP and let $G\not\in (c,k)$-LCP.
Then $G$ has a $K_{1,\leq \max\{3,k\}}$-cover.
\end{lemma}
\begin{proof} Let $G$ be such a reduced graph. We first show that $G$ has a star cover. Since it is not possible to apply $R_{0,k}$, $G$ has no isolated vertex. By choosing a spanning tree of each component of $G$, we obtain a forest $F$. If a tree in $F$ is not a star, it has an edge not incident to a leaf. As long as $F$ contains such an edge, delete it from $F$. Observe that $F$ becomes a star cover of $G$. However, the number of leaves in each star of $F$ is only bounded by $\Delta(G)$. We will show that among the possible star covers of $G$, there exists a $K_{1,\leq \max(3,k)}$-cover.

For each star cover $F$, we define the $F$-{\em sequence} ($n_{F,\Delta(G)}$,$n_{F,\Delta(G)-1}$,$\ldots$, $n_{F,1}$), where $n_{F,i}$ is the number of stars with exactly $i$ edges, $i \in [\Delta(G)]$. We say a star cover $F_{1}$ is {\em smaller} than a star cover $F_{2}$ if and only if the $F_{1}$-{\em sequence} is smaller than the $F_{2}$-{\em sequence} lexicographically, i.e. there exists some $i \in [\Delta(G)]$ such that $n_{F_{1}, i} < n_{F_{2}, i}$ and for every $j > i$, $n_{F_{1}, j} = n_{F_{2}, j}$. We select a star cover $S$ which has the lexicographically minimum sequence, that is, for any star cover $F \neq S$ of $G$, the $S$-sequence is smaller or equal to the $F$-sequence. 
Suppose that $\Delta(S) > \max\{3,k\}$. Let $C_i$ ($L_i$, respectively) be the set of all the centers (leaves, respectively) of all stars of $S$ isomorphic to $K_{1, i}$. We also define $L_{\geq i}=\cup_{j\geq i}L_j$. We will now prove two claims.

\paragraph{Claim 1} {\em There is no edge $uv \in E(G)\setminus E(S)$ such that $u \in L_{\geq 3}$ and $v \in L_{\geq 1}$.} 

Indeed, suppose there exists one and let 
$x,y$ be such that $xu \in E(S)$, $yv \in E(S)$.
If $v \in L_{\geq 2}$, then by deleting edges $xu, yv$ and adding edge $uv$, we do not create any isolated vertex but we decrease the size of the stars centered at $x$ and $y$, and thus we get a smaller star cover than $S$, a contradiction. Otherwise, $v$ is an endvertex of an independent edge, and by deleting edge $xu$ and adding edge $uv$, we decrease the size of the star centered at $x$, and create a star $K_{1,2}$ centered at $v$, which still induces a star cover smaller than $S$, a contradiction. 

\paragraph{Claim 2} {\em Suppose $S$ contains a star isomorphic to $K_{1, i}$ and centered at vertex $x$, and a star isomorphic to $K_{1, j}$ and centered at vertex $y$, such that $i-j\ge 2$. There is no path from $x$ to $y$ in which the odd edges are in $E(S)$ and go from a center to a leaf, and the even edges are in $E(G)\setminus E(S)$ and go from a leaf to a center.}

Suppose there exists such a path. Then by deleting the odd edges of the path and adding the even ones, we do not create isolated vertices because $x$ still has leaf-neighbors, $y$ gets a neighbor, every transitional center keeps the same number of leaf-neighbors and the transitional leaves always go to a new center. This operation only decreases the size of star centered at $x$ by 1 and increases the size of star centered at $y$ by 1, giving us a lexicographically smaller star cover, a contradiction.\\

Now, let $S'$ be the subgraph of $S$ containing all stars $K_{1,\Delta(S)}$ of $S$. While there is an edge $uv \in E(G)\setminus E(S)$ such that $u$ is a leaf of $S'$ and $v \in C_{\Delta(S)-1}\setminus S'$, we add the star centered at $v$ to $S'$. This procedure terminates because  $C_{\Delta(S)-1}$ is finite.

Let $C'$ ($L'$, respectively) be the centers (leaves, respectively) in $S'$. Assume now there is an edge $uv \in E(G)\setminus E(S)$ such that $u \in L' \subseteq L_{\ge \Delta(S)-1} \subseteq L_{\ge 3}$ and $v \in V(G)\setminus C'$. By Claim 1, $v \not\in L_{\ge 1}$. Since $v \not\in C_{\Delta(S)} \subseteq C'$ and since the above procedure has terminated, $v \in C_j$ for some $j$ such that $\Delta(S) -j \ge 2$. Now, by construction, there is a alternating path from a vertex in $C_{\Delta(S)}$ to a vertex in $C_j$ of the type described in Claim 2, which is impossible. 

So, there is no edge $uv \in E(G)\setminus E(S)$ such that $u \in L'$ and $v \not\in C'$. This means that for any $u \in L', N(u) \subseteq C'$. Furthermore, for each $u \in C'$, we can
define $V_u$ to be the leaves of the star centered at $u$, for which we have $|V_u| \ge  \Delta(S)-1 \geq k$.
So, $(C',L')$ is an obstacle from $O_{|C'|,k}$. Since $G$ is reduced for $(c,k)$-LCP, $|C'| \ge c$ and thus $G[S'] \in (|C'|,k)$-LCP. This implies that $G \in (c,k)$-LCP, a contradiction. 
\end{proof}

Now we can prove the following:

\begin{theorem}\label{th:1}
For every fixed $c$, if $G$ is reduced for $(c,k)$-LCP and has at least $2ck$ vertices, then $G \in (c,k)$-LCP. 
Thus, $(c,k)$-{\sc Load Coloring} admits a kernel with less than $2ck$ vertices. 
\end{theorem}
\begin{proof}
Observe that for every $c$, $G\in (c,0)$-LCP, and $G\in (c,1)$-LCP if and only if $G$ has a matching with at least $c$ edges. Thus, we may assume that $k\ge 2$. By Lemmas \ref{lem2} and \ref{poly}, we can map, in polynomial time, any instance $(G,c)$ into an instance $(G',c')$ such that $c' \le c$ and $G'$ is reduced for $(c',k)$-LCP. We therefore may assume that $G$ is reduced for $(c,k)$-LCP. Suppose that $G \not \in (c,k)$-LCP and $n \ge 2ck$. By Lemma \ref{lem:coverExists}, $G$ has a $K_{1,\leq \max(3,k)}$-cover which is a $K_{1,< 2k}$-cover, since we assumed $k \ge 2$. But then, Lemma \ref{lem:starkernel} implies that $G \in (c,k)$-LCP, a contradiction. 
\end{proof}

\section{Bounding Number of Edges in Kernel}\label{sec:ek}

In the previous section, we proved that $(c,k)$-{\sc Load Coloring} admits a kernel with less than $2ck$ vertices. We would like to bound the number of edges in a kernel for the problem.

\begin{lemma}\label{lem:bi}
Let $b(c,k,n) = c^2k + n(c-1)$. For every integer $i\ge 0$ and bipartite graph $G$ with $n$ vertices, if $m \ge b(2^i,k,n)$ then $G \in (2^i,k)$-LCP.
\end{lemma}
\begin{proof}
We will prove the lemma by induction on $i$. For the base case, observe that any graph with at least $k = b(1,k,n)$ edges is in $(1,k)$-LCP for every $k$ and $n$. We now assume the claim holds for any $j$ smaller than $i+1$ and want to prove it for $i+1$. Consider a bipartite graph $G = (A \cup B, E)$ with $n$ vertices such that $G\not\in (2^{i+1},k)$-LCP. Let $A_1 = A$, $B_1 = B$ and $A_2 = B_2 = \emptyset$. While there exists $u \in B_1$ such that $|E(A,B_2 \cup \{u\})| < b(2^i,k,|A|+|B_2 \cup \{u\}|) + b(2^i,k,|B_2 \cup \{u\}|)$, move $u$ from $B_1$ to $B_2$. So now assume there is no such $u$. Then, while $|E(A_1,B_1)| \ge b(2^i,k,|A_1|+|B_1|) + |A_1|$ and $|E(A_2,B_1)| < b(2^i,k,|A_2|+|B_1|) + |A_2|$, move an arbitrary vertex from $A_1$ to $A_2$. Since we only move vertices from $A_1$ to $A_2$ or from $B_1$ to $B_2$, we always have $A = A_1 \cup A_2$ and $B = B_1 \cup B_2$. Eventually, the partition of $A\cup B$ falls into one of two cases: 
\begin{itemize}
\item $|E(A_1,B_1)| < b(2^i,k,|A_1|+|B_1|) + |A_1|$. If $A_2 = \emptyset$,  then $|E(A_2,B_1)| = 0$. Otherwise, let $v$ be the last vertex moved from $A_1$ to $A_2$. Observe that $|E(A_2,B_1)| \le |E(A_2\setminus \{v\},B_1)| + |B_1| < b(2^i,k,|A_2\setminus \{v\}|+|B_1|) + |A_2\setminus \{v\}| + |B_1|$. In both cases, $|E(A_2,B_1)| < b(2^i,k,|A_2|+|B_1|)+|A_2|+|B_1|$. Thus, we have $|E(G)| = |E(A_1,B_1)| + |E(A_2,B_1)| + |E(A,B_2)| < (b(2^i,k,|A_1|+|B_1|) + |A_1|) + (b(2^i,k,|A_2|+|B_1|) + |A_2| + |B_1|) + (b(2^i,k,|A|+|B_2|) + b(2^i,k,|B_2|)) \le 4(2^{2i})k + 2n(2^i-1) + n = 2^{2(i+1)}k + n(2^{i+1}-1) = b(2^{i+1},k,n)$, as required.

\item $|E(A_1,B_1)| \ge b(2^i,k,|A_1|+|B_1|) + |A_1|$. In this case, we also have $|E(A_2,B_1)| \ge b(2^i,k,|A_2|+|B_1|) + |A_2|$. Let $u$ be an arbitrary vertex in $B_1$. Observe that $|E(A_1,B_1\setminus \{u\})| \ge b(2^i,k,|A_1|+|B_1|)$ and $|E(A_2,B_1\setminus \{u\})| \ge b(2^i,k,|A_2|+|B_1|)$. We also have $|E(A,B_2 \cup \{u\})| \ge b(2^i,k,|A|+|B_2 \cup \{u\}| ) + b(2^i,k,|B_2 \cup \{u\}|)$. It is not possible that $|E(A_1,B_2 \cup \{u\})| < b(2^i,k,|A_1|+|B_2 \cup \{u\}|)$ and $|E(A_2,B_2 \cup \{u\})| < b(2^i,k,|A_2|+|B_2 \cup \{u\}|)$ as otherwise, $|E(A,B_2 \cup \{u\})| = |E(A_1,B_2 \cup \{u\})| + |E(A_2,B_2 \cup \{u\})| < b(2^i,k,|A_1|+|B_2 \cup \{u\}|) + b(2^i,k,|A_2|+|B_2 \cup \{u\}|) = b(2^i,k,|A|+|B_2 \cup \{u\}|) + b(2^i,k,|B_2 \cup \{u\}|)$. So, there exist disjoint vertex sets $X$ and $Y$ such that $|E(X)| \ge b(2^i,k,|X|)$ and $|E(Y)| \ge b(2^i, k, |Y|)$
(either $X = A_1 \cup B_1 \setminus \{u\}$ and $Y = A_2 \cup B_2 \cup \{u\}$, or $X = A_2 \cup B_1 \setminus \{u\}$ and $Y = A_1 \cup B_2 \cup \{u\}$). 
Thus, by taking a suitable $2^i$-coloring of $X$ and a suitable $2^i$-coloring of $Y$, we have that $G \in (2^{i+1},k)$-LCP, a contradiction.

\end{itemize} 
So we have proved that the claim also holds when $j=i+1$, i.e. if $m \geq b(2^{i+1},k,n)$ then $G\in(2^{i+1},k)$-LCP.
\end{proof}

\begin{lemma}\label{lemm6}
Let $f(c,k,n) = (2c-1)ck+2n(c - 1)$. For every nonnegative integer $i$ and every graph $G$ with $n$ vertices, if $m \ge f(2^i,k,n)$ then $G \in (2^i,k)$-LCP.
\end{lemma}
\begin{proof}
We will prove the lemma by induction on $i$. For the base case, observe that any graph with at least $k = f(1,k,n)$ edges is in $(1,k)$-LCP for every $k$ and $n$. We now assume the claim holds for any $j$ smaller than $i+1$ and want to prove it for $i+1$. Consider a graph $G$ with $n$ vertices such that $G\not\in (2^{i+1},k)$-LCP and $|E(G)| \ge f(2^i,k,n)$. 

We will show that there exists a set $A \subseteq V(G)$ such that $f(2^i,k,|A|) \le |E(A)| \le f(2^i,k,|A|)+|A|$ (and thus $G[A] \in (2^i,k)$-LCP). We may construct the set $A$ as follows: initially $A = \emptyset$ and while $|E(A)| < f(2^i,k,|A|)$, add an arbitrary vertex of $V(G)\setminus A$ to $A$. Let $u$ be the last added vertex; we have $f(2^i,k,|A|) \le |E(A)| \le |E(A\setminus \{u\})| + |A\setminus \{u\}| < f(2^i,k,|A\setminus \{u\}|) + |A\setminus \{u\}| < f(2^i,k,|A|) + |A|$. 

Let $B = V(G)\setminus A$. If $G[B] \in (2^i,k)$-LCP, then $G \in (2^{i+1},k)$-LCP, a contradiction. So $|E(B)| < f(2^i,k,|B|)$. Furthermore, $|E(A,B)| < b(2^{i+1},k,n)$, as otherwise we are done by Lemma \ref{lem:bi}. Finally, $|E(G)| = |E(A)| + |E(B)| + |E(A,B)| < f(2^i,k,|A|) +  f(2^i,k,|B|) + n + b(2^{i+1},k,n) = f(2^{i+1},k,n)$. The claim holds when $j=i+1$, which completes the proof. 
\end{proof}

\begin{theorem}\label{thm:edgeKernel}
The $(c,k)$-{\sc Load Coloring} Problem admits a kernel with less than
$f(2c,k,2ck) = 16c^{2}k-6ck$ edges.
\end{theorem}
\begin{proof}
By Theorem \ref{th:1}, we can get a kernel with less than $2ck$ vertices.
Let $c'$ be the minimum power of 2 such that $c\leq c'$. Observe that $c'<2c$ and thus by Lemma \ref{lemm6}
we get a kernel with $|E(G)|\le f(c',k,2ck)< f(2c,k,2ck)=16c^{2}k-6ck$. 
\end{proof}

We now consider an approximation algorithm for the $c$-{\sc Load Coloring} Problem: Given a graph $G$ and integer $c$, we wish to determine the maximum $k$, denoted $k_{opt}$, for which $G\in (c,k)$-LCP. 
We define the approximation ratio  $r(c)=\frac{k_{opt}}{k}$, where $k$ is the output of the approximation algorithm.

Let $K(c)k$
 be an upper bound of the number of edges in a kernel for $(c,k)$-{\sc Load Coloring} and
let $P(c)=\prod_{i=1}^c\frac{K(i)}{i}$. For $c=1$, we may assume that $K(1)=1$ as $(1,k)$-{\sc Load Coloring} is trivially polynomial time solvable. Hence $P(1)=1.$
For $c \ge 2$, we have $K(c) = 16c^2 - 6c$.

\begin{theorem}\label{theo:app} There is a $2^{c-1} P(c)$-approximation algorithm for $c$-{\sc Load Coloring}. 
\end{theorem}

\begin{proof}
We prove the claim by induction on $c$.
For $c = 1$, we have $P(1) = 1$.
Assume the lemma is true for all $c' < c$.

Let $G$ be an instance for $c$-{\sc Load Coloring} with $n$ vertices and $m$ edges. We may assume that $G$ has no isolated vertices. Clearly, $k_{opt}\le \frac{m}{c}$. Consider $k = \lfloor \frac{m}{K(c)} \rfloor$. 

If $k=0$, then $m<K(c)$ and we can find $k_{opt}$ in $O(1)$ time.

Now let $k > 0$. If $n \le 2ck$, then by the proof of Theorem \ref{thm:edgeKernel}, since $m \ge K(c)k$, $G\in (c,k)$-LCP. 
So we return $k$, and $\frac{k_{opt}}{k} \le \frac{m}{ck} \le \frac{K(c)(k+1)}{ck} \le \frac{2K(c)}{c} \le 2^{c-1}P(c)$.

If $n \ge 2ck$ and $G$ is reduced for $(c,k)$-LCP, then by Theorem \ref{th:1}, $G\in (c,k)$-LCP and we return $k$ as above.
If $n \ge 2ck$ and $G$ is not reduced for $(c,k)$-LCP, we can use Lemma \ref{poly} to reduce $(G,c)$ to $(G',c')$ with $c' < c$.
By induction we may find $k'$ such that $k'_{opt} \leq 2^{c'-1} P(c') k',$ where $k'_{opt}$ is the optimal solution for $c'$-{\sc Load Coloring} on $G'.$
Now consider three cases.

\begin{itemize}

\item $k' \ge k$. Then $G' \in (c',k)$-LCP and so $G \in (c,k)$-LCP. This is also a Yes-Instance case which leads to the same conclusion.

\item $k'_{opt} \le 2^{c'-1} P(c') k' < k$. Because $k'_{opt}+1 \leq 
k$, an obstacle from $O_{c-c',k}$ is also an obstacle from 
$O_{c-c',k'_{opt}+1}$, therefore $G'$ can be derived from $G$ using a
reduction rule for $(c,k'_{opt}+1)$-LCP.
Since $G' \not \in (c',k'_{opt}+1)$-LCP, $G \not \in 
(c,k'_{opt}+1)$-LCP.
Thus $k_{opt} = k'_{opt}$. The algorithm may output $k'$ which 
satisfies $k_{opt} = k'_{opt} \leq 2^{c'-1} P(c') k' \le 2^{c-1} P(c) k$.

\item $k' < k \leq 2^{c'-1} P(c')k'$. The algorithm gives $k'$ as an approximation of $k_{opt}$. Then $\frac{k_{opt}}{k'} \le \frac{m}{ck'} \le \frac{K(c)(k+1)}{ck'} \le  \frac{K(c)}{c} \frac{2k}{k'} \le \frac{K(c)}{c} 2^{c'} P(c') \le 2^{c-1} P(c)$.

\end{itemize}

In every case, the approximation ratio is at most $2^{c-1} P(c)$.
\end{proof}

\section{Number of Edges in Kernel for $c = 2$}\label{sec:ek2}

In this section, we look into the edge kernel problem for the special case when $c=2$. By doing a refined analysis, we will give a kernel with less than $8k$ edges for $(2,k)$-LCP, which is a better bound than the general one.
Henceforth, we assume that $G$ is reduced for $(2,k)$-LCP, and just consider the case when $|V(G)|<4k$, as we have proved that if $|V(G)|\ge 4k$ then $G\in (2,k)$-LCP.
\begin{lemma}\label{lem:smallComponents}
If $G$ has at least $3k-2$ edges and  every component in $G$ has less than $k$ edges then $G \in (2,k)$-LCP.
\end{lemma}
\begin{proof}
We consider colorings of the graph such that vertices in the same component are colored with the same color. Thus every edge in the graph is colored with  1 or  2. Denote the set of edges colored $i$ with $E_{i}, i=1,2.$ Among all possible colorings, choose a coloring of the graph such that $|E_1| \geq |E_2|$ and $||E_{1}|-|E_{2}||$ is minimum.
Suppose $|E_{2}| \leq k-1$, then $|E_{1}| \geq 2k-1$, $||E_{1}|-|E_{2}||>k$.  Changing the color of one component from 1 to 2, we get a new coloring of the graph. For the new coloring, denote the set of edges colored $i$ with $E_{i}', i=1,2$. Since each component has less than $k$ edges, $|E_{1}|>|E_{1}'|\geq k, |E_{2}'|\leq 2k-2$. So $||E_{1}'|-|E_{2}'||<||E_{1}|-|E_{2}||$, a contradiction. Therefore we have $|E_{1}| \geq |E_{2}|\geq k$, so $G \in (2,k)$-LCP.
\end{proof}

If $G$ has at least two components, each with at least $k$ edges, it is obviously a Yes-instance.  
 Therefore by Lemma \ref{lem:smallComponents}, we may assume there is exactly one component $C$ with at least $k$ edges in the graph. Denote the total number of edges in $G-V(C)$  with $m'$.
Observe that if $m' \geq k$, trivially $G \in (2,k)$-LCP. So assume that $m' < k$.

\begin{lemma}
 If $G$ is a reduced graph for $(2,k)$-LCP, $m' < k$ and $\Delta=\Delta(G) \geq3k-2m'$, then $G \in (2,k)$-LCP.
\end{lemma}

\begin{proof}
Let $u$ be one of the vertices with degree $\Delta$ and $N(u)$ its neighbors. Because the graph is reduced by Reduction Rule $R_{1,k}$, $u$ has at least $2k-2m'$ neighbors which are not leaves. Arbitrarily select $k-m'$ vertices among them and for each one, select any neighbor but $u$. Color the selected vertices and $G-V(C)$  by 1. By construction, there are at least $k$ edges  colored 1 and there are at most $2k-2m'$ colored vertices in $N(u)$. So there are at least $k$ uncolored vertices in $N(u)$. We color them and $u$ with 2. So $G \in (2,k)$-LCP.
\end{proof}

The next lemma deals with the case $\Delta=\Delta(G)<3k$.  

\begin{lemma}
Let $G$ be a graph with $\Delta<3k$ and $|E(G)| \ge 8k$, then  $G \in (2,k)$-LCP. 
\end{lemma}

\begin{proof}

Because of Lemma \ref{lem:smallComponents}, we may assume there exists a connected component $C$ with at least $k$ edges. In this component, choose a minimal set $A \subseteq V(C)$ such that $|A|\leq k+1$ and $|E(A)|= k+d\geq k$. We may find such a set $A$ in the following way. Select arbitrarily a vertex in $C$ and put it into $A$, then keep adding to this set some neighbor of some vertex in $A$ until $|E(A)| = k+d \geq k$. Since each time we select a neighbor of $A$ we strictly increase $|E(A)|$, $|A| \le k+1$. If there is any vertex $u \in A$ with $|N_{A}(u)|\leq d$, then $A'=A\setminus \{u\}$ is a smaller vertex set such that $|E(A')|\geq k$. Thus, we may remove such vertices until $|E(A)| = k+d$ and for each vertex $u \in A$, $|N_A(u)|>d$. Denote $B=V(G)\setminus A$. We may assume $|E(B)|<k$, as otherwise $G \in (2,k)$-LCP.

We now show that $|A|+d\leq k+3$. Since every vertex $u \in A$ has $d_{A}(u) > d$, $|E(A)| = \frac{1}{2}\Sigma_{u \in A}d_{A}(u)\geq \frac{d+1}{2}|A|$. We have $k+d =|E(A)| \geq \frac{d+1}{2}|A|$, thus $|A| \leq \frac{2(k+d)}{d+1}$. Moreover as $d \leq |A|-1$, $$d+|A| \leq 2|A|-1 \leq \frac{4(k+d)}{d+1}-1 < \frac{4k}{d+1}+3$$ If $d\geq 3$, we have our result, otherwise $d \le 2$ and $d +|A| \le 2+k+1 = k+3$.

Let $A_1,A_2,B_1,B_2$ be a partition of $V(G)$ such that
$A = A_1 \cup A_2$, $B = B_1 \cup B_2$, $|A_2| = 1$ and $|E(A, B_2)| < 2k$. Such a partition is possible: let $y = \argmax\{|N_{B}(u)|: u \in A\}$ and initially take $A_1 = A\setminus \{y\}, A_2 = \{y\}, B_1=B, B_2 = \emptyset$. Suppose $|E(A_1,B_1)| \le k + |A_1|$ then $|E(G)| \le |E(A)| + |E(B)| + |E(A_1,B_1)| + |E(A_2,B_1)| \le (k+d)+(k-1)+(k+|A|-1)+\Delta \le 7k+1$, a contradiction since $|E(G)| > 8k$. So, $|E(A_1,B_1)| > k + |A_1|$. We will consider two cases:  $\max \{|N_{B_1}(u)|:\ u \in A\}$ is greater than $k$ or not.

If so, observe that $|E(A_2,B_1)| = |E(\{y\},B_1)| = \max \{|N_{B_1}(u)| : u \in A\} > k$. Move all vertices of $B_1\setminus N(y)$ to $B_2$. We still have $|E(\{y\},B_1)| > k$ and $|E(\{y\},B_2)| = 0$. Moreover $B_1 \subseteq N(y)$. If $|E(A_1,B_2)| \ge k$, then $G$ is in $(2,k)$-LCP, thus $|E(A_1,B_2)| < k$. While $|E(\{y\},B_1)| \ge k+1$ and $|E(A_1,B_1)| \ge k+|A_1|$, move an arbitrary vertex from $B_1$ to $B_2$. After each move, $|E(\{y\},B_1)| \ge k$ and $|E(A_1,B_1)| \ge k$, thus $|E(A_2,B_2)| < k$ and $|E(A_1,B_2)| < k$ as otherwise, $G$ would be in $(2,k)$-LCP.

Eventually, we have $|E(A_1,B_1)| < k+|A_1|$ or $|E(\{y\},B_1)| = k$. Suppose $|E(A_1,B_1)| < k+|A_1|$, then $|E(G)| \le |E(A)| + |E(B)| + |E(A_1,B_1)| + |E(A_1,B_2)| + |E(\{y\},B)| \le (k+d) + (k-1) + (k+|A_1|-1) + (k-1) + \Delta \le 4k-3+(d+|A|)+\Delta < 8k$, a contradiction. Thus, $|E(A_1,B_1)| \ge k+|A_1|$ and $|E(\{y\},B_1)| = k$. As $B_1 \subseteq N(y)$, we have $|B_1| = k$. We have found a new partition with the wanted properties and with $\max \{|N_{B_1}(u)|:\ u \in A\} \le |B_1| = k$.

We now study the case $\max \{|N_{B_1}(u)| : u \in A\} \le k$. While there exists $u \in B_1$ such that $|E(A, B_2 \cup \{u\})| < 2k$, move $u$ from $B_1$ to $B_2$. 
Then, (if and) while $|E(A_1,B_1)| \ge k + |A_1|$ and $|E(A_2,B_1)| < k + |A_2|$, move an arbitrary vertex from $A_1$ to $A_2$. 

After all such moves, suppose that $|E(A_1,B_1)| < k+|A_1|$. If $|A_2| = 1$, we have $|E(A_2,B_1)| \le \max \{|N_{B_1}(u)| : u \in A\} \le k$, otherwise we moved some vertices from $A_1$ to $A_2$. Let $u$ be the last one. Since $|E(A_2\setminus \{u\},B_1)| < k + |A_2\setminus \{u\}|$, we know $|E(A_2,B_1)| \leq |E(A_2\setminus \{u\},B_1)| + \max \{|N_{B_1}(u)|:\ u \in A\} < k+|A_2|-1+k = 2k+|A_2|-1$. For both cases, $|E(G)| = |E(A)| + |E(B)| + |E(A_1,B_1)| + |E(A_2,B_1)| + |E(A,B_2)| \le (k+d)+(k-1)+(k+|A_1|-1)+(2k+|A_2|-2)+(2k-1) \le 7k + d + |A| -5 < 8k$, which is impossible. 

So, $|E(A_1,B_1)| \ge k+|A_1|$ which implies $|E(A_2,B_1)| \ge k+|A_2|$. For any vertex $u \in B_1$, we have $|E(A_1,B_1\setminus \{u\})| \ge k$ and $|E(A_2,B_1\setminus \{u\})| \ge k$ and we also obtain $|E(A, B_2 \cup \{u\})| \ge 2k$, i.e $E(A_1,B_2 \cup \{u\})$ or $E(A_2,B_2 \cup \{u\})$ has at least $k$ edges. Thus $G\in (2,k)$-LCP.
\end{proof}

The lemmas of this section and the fact that their proofs can be turned into polynomial algorithms, imply the following:

\begin{theorem}\label{thm3}
If $G$ is reduced for $(2,k)$-LCP and has at least $8k$ edges, then $G \in (2,k)$-LCP. Thus, $(2,k)$-{\sc Load Coloring} admits a kernel with less than $8k$ edges.
\end{theorem}

Since we have a better bound for the number of edges in a kernel when $c=2$, we may get a better approximation when $c=2$.

\begin{theorem}
For every $\varepsilon>0$, there is a $(4+\varepsilon)$-approximation algorithm for $2$-{\sc Load Coloring}.
\end{theorem}

\begin{proof}
Let $G$ be an instance for $2$-{\sc Load Coloring} with $m = 8p+q$ edges, where $0 \le q < 8$. Let $k_{opt}$ be the optimal solution of $2$-{\sc Load Coloring} on $G$, and observe that  $k_{opt}\le \lfloor \frac{m}{2} \rfloor \le 4p+3$. Let $p_0=\lceil \frac{3}{\varepsilon}\rceil$.
If $p\le p_0-1$ then we can find $k_{opt}$ in $O(1)$ time.


So assume that $p \geq p_0$. Note that $\frac{k_{opt}}{p}\leq \frac{4p+3}{p}\leq 4+\varepsilon$. If $G$ is reduced for $(2,p)$-LCP, $G \in (2,p)$-LCP by Theorem \ref{thm3}, and so $p$ gives the required approximation.  We may assume that $G$ is not reduced for $(2,p)$-LCP and reduce $G$ to $G'$.  If $|E(G')| \ge p$, then $G' \in (1,p)$-LCP, and by Lemma \ref{lem2}, $G \in (2,p)$-LCP. Again, $p$ gives the required approximation.

Now assume that $|E(G')| < p$ and let $k'_{opt} = |E(G')|$ be the optimal solution of $1$-{\sc Load Coloring} on $G'$. Then $k'_{opt}+1 \leq p$ and so an obstacle from $O_{1,p}$ is also an obstacle from $O_{1,k'_{opt}+1}$. Therefore, $G'$ can be derived from $G$ using a reduction rule for $(2,k'_{opt}+1)$-LCP. Since $G' \not \in (1,k'_{opt}+1)$-LCP, $G \not \in (2,k'_{opt}+1)$-LCP.
Thus $k_{opt} = k'_{opt} = |E(G')|$. So let our algorithm output $|E(G')|$ in this case. 
\end{proof}

\section{Discussions}\label{sec:d}

In the {\sc Judicious Bipartition} Problem (see, e.g., the survey \cite{S2005}), given a graph $G$, we are asked to find a bipartition $V_1,V_2$ of $V(G)$ which minimizes $\max\{|E(V_1)|,|E(V_2)|\}$. To see that {\sc Judicious Bipartition} and $2$-{\sc Load Coloring} are different problems, following \cite{NABA2007} consider $2nK_2$, the union of $2n$ disjoint edges, and observe that while the solution of $2$-{\sc Load Coloring} is $n$, that of {\sc Judicious Bipartition} is zero. 

To the best of our knowlege, we obtained the first linear-edge kernel for a nontrivial problem on general graphs. As we could see, such kernels can be used to obtain approximation algorithms. It would be interesting to obtain such kernels for other nontrivial problems.




\begin{thebibliography}{1}
\bibitem {NABA2007} N. Ahuja, A. Baltz, B. Doerr, A. Privtivy and A. Srivastav, On the Minimum Load Coloring Problem. J. Discr. Alg. 5(3): 533-545 (2007).
\bibitem{meta} H. L. Bodlaender, F. V. Fomin, D. Lokshtanov, E. Penninkx, S. Saurabh, and D. M.
Thilikos. (Meta) Kernelization. In FOCS 2009, pp 629-638, 2009.

\bibitem{CFLMPPS2015} M. Cygan, F.V. Fomin, L. Kowalik, D. Lokshtanov, D. Marx, M. Pilipczuk, M. Pilipczuk, and S. Saurabh,  Parameterized Algorithms, Springer, 2015, in press.
\bibitem {DF2013} R.G. Downey and M.R. Fellows, Foundations of Parameterized Complexity, Springer, 2013.
\bibitem {JM2006} J. Flum and M. Grohe, Parameterized Complexity Theory, Springer, 2006.
\bibitem {GJ2013} G. Gutin and M. Jones, Parameterized Algorithms for Load Coloring Problem, Inform. Proc. Lett. 114:446-449, 2014.
\bibitem{LMS2012} D. Lokshtanov, N. Misra, and S. Saurabh, Kernelization - preprocessing with a guarantee. In {\em The Multivariate Algorithmic Revolution and Beyond}, LNCS 7370:129-161, 2012.
\bibitem {R2006} R. Niedermeier, Invitation to Fixed-Parameter Algorithms. Oxford UP, 2006.
\bibitem{S2005} A.D. Scott, Judicious partitions and related problems, in {\em Surveys in Combinatorics 2005}, London Math. Soc. Lect. Note Ser. 327:95-117, 2005.
\end{thebibliography}
\end{document}